\documentclass[aps,pra,reprint,superscriptaddress,showpacs]{revtex4-1}
\usepackage{mathrsfs}
\usepackage{amssymb, amsbsy, amsmath, amsthm, latexsym, dsfont, array, layout, graphicx, mathrsfs, color, ragged2e,soul}
\usepackage[colorlinks=true,citecolor=blue,linkcolor=blue,anchorcolor=blue,urlcolor=blue]{hyperref}
\hypersetup{linktoc=all}

\allowdisplaybreaks

\newcommand{\ketbra}[2]{\left|{#1}\rangle\!\langle{#2}\right|}

\newcommand{\black}{\color[rgb]{0,0,0}}
\renewcommand{\l}{\left}
\renewcommand{\r}{\right}
\usepackage{braket}

\renewcommand{\url}[1]{}
\newcommand{\urlprefix}{}
\renewcommand{\href}[1]{}

\begin{document}

\title{Qudit hypergraph states and their properties}
\author{Fei-Lei Xiong}
\author{Yi-Zheng Zhen}
\author{Wen-Fei Cao}
\author{Kai Chen}
\author{Zeng-Bing Chen}
\email{zbchen@ustc.edu.cn}

\affiliation{Hefei National Laboratory for Physical Sciences at Microscale and Department of Modern Physics, University of Science and Technology of China, Hefei, Anhui 230026, China\\}
\affiliation{CAS Center for Excellence and Synergetic Innovation Center of Quantum Information and Quantum Physics, University of Science and Technology of China, Hefei, Anhui 230026, China\\}

\begin{abstract}
Hypergraph states, a generalization of graph states, constitute a large class of quantum states with intriguing non-local properties and have promising applications in quantum information science and technology. In this paper, we generalize hypergraph states to qudit hypergraph states, i.e., each vertex in the generalized hypergraph (multi-hypergraph) represents a $d$-level quantum system instead of a qubit. It is shown that multi-hypergraphs and $d$-level hypergraph states have a one-to-one correspondence. We prove that if one part of a multi-hypergraph is connected with the other part, the corresponding subsystems are entangled. More generally, the structure of a multi-hypergraph reveals the entanglement property of the corresponding quantum state. Furthermore, we discuss their relationship with some well-known state classes, e.g., real equally weighted states and stabilizer states. These states' responses to the generalized $Z$ ($X$) operations and $Z$ ($X$) measurements are studied. The Bell non-locality, an important resource in fulfilling many quantum information tasks, is also investigated.
\end{abstract}

%%%%%ATTENTION%%%%

%\noindent
%\blue{Blue - revision by XFL}\\
%\red{Red - revision by ZYZ}\\
%\green{Green - revision by CWF}\\
%\black
%\textbf{Bold - ATTENTION}

%%%%%%%%%%%%%%%%%%

\maketitle
\section{INTRODUCTION}
\label{sec:1}
In quantum information science and technology, graph states constitute  an almost \black unique family of states for their appealing properties and applications~\cite{Raussendorf01,Raussendorf03,Hein04,Aschauer05,Schlingemann01,Schlingemann02,Schlingemann03,Zhou03,Hall07,Sean06,Looi08}. They can be used to implement one-way quantum computation~\cite{Raussendorf01} and construct quantum codes~\cite{Schlingemann01,Schlingemann02,Schlingemann03}. Moreover, they can be used to characterize many kinds of widely used entangled states, such as cluster states~\cite{Briegel01},  the \black Greenberger-Horne-Zeilinger (GHZ) states~\cite{greenberger1990bell} and more generally, stabilizer states~\cite{Gottesman96, Gottesman98}. To make quantum states of suitable physical systems describable in the framework as that of graph states, Ref.~\cite{Ionicioiu12} introduced an axiomatic method. Later, Ref.~\cite{Qu13, Rossi13} generalized this approach and introduced a new class of quantum states named hypergraph states.

Like graph states, given a hypergraph, one can define  an associated \black qubit hypergraph state, i.e., hypergraphs can be encoded into quantum states \cite{Rossi13, Qu13}. Besides this feature, every qubit hypergraph state corresponds to a stabilizer group~\cite{Gottesman96, Gottesman98}. However, generally speaking, the stabilizers are no longer products of local operators \cite{Rossi13}. As a new class of quantum states, they possess lots of new properties, e.g., local unitary symmetries~\cite{QuMa13,Chen14,David15,Guhne14}, entanglement properties~\cite{QuLi13,Qu14,Qu15,Guhne14} and non-local properties~\cite{Guhne2005,Toth2006,Guhne14,Gachechiladze16}. Besides these fundamental properties, these states also have many applications. Qubit hypergraph states are \emph{real equally weighted state}s \cite{Bru11,Qu12},
which have important applications in Grover~\cite{Grover96} and Deutsch-Jozsa~\cite{Deutsch92} algorithms. Recently, Ref.~\cite{Makmal14} has shown that, if one has a black box that can tell whether an input qubit hypergraph state is a product state, he/she can solve the NP-complete SAT problem efficiently~\cite{gary1979computers}. Fully connected $k$-uniform qubit hypergraph states, a generalization of GHZ states, are applicable in Heisenberg-limited quantum metrology with more robustness to noise and particle losses~\cite{Rossi13, Gachechiladze16}. Superior to one-way quantum computation based on graph states, measurement-based quantum computation with qubit hypergraph states is non-adaptive, making the measurement scheme simpler~\cite{hoban2011non}.

In this paper, we implement the concept of multi-hypergraphs \cite{berge1973graphs} and encode these multi-hypergraphs into multi-qudit quantum states. We investigate the relationship between the multi-hypergraphs and their corresponding qudit hypergraph states, mainly about the map from the set of multi-hypergraphs to the set of qudit hypergraph states, and the relationship between the connectivity and entanglement. We investigate the local unitary transformations and local measurements on qudit hypergraph states, for  pushing \black their potential applications in error-correcting quantum codes and quantum computation. The Bell non-locality, a useful resource in quantum computation and quantum high precision measurement, is also studied. Furthermore, a systematic approach for the experimental detection is provided.

The paper is organized as follows: In Sec.~\ref{sec:2}, we give some preliminary knowledge of hypergraph and qubit hypergraph states, and explain related terminologies. We then generalize these concepts to represent a larger class of quantum states, which we call qudit hypergraph states, using a similar formalism. We show how the notion hypergraph should be modified when each vertex represents a qudit instead of a qubit. In Sec.~\ref{sec:3}, we discuss  the \black connection between the generalized hypergraphs and the qudit hypergraph states, including the correspondence and relationship between the connectivity and the entangled properties.  In Sec. \ref{sec:?}, we discuss the relationship among qudit hypergraph states and some well-known state classes, like real equally weighted states, qudit graph states and stabilizer states. \black In Sec.~\ref{sec:4}, qudit hypergraph states' responses to local $Z$ and $X$ operations (measurements) are investigated. In Sec.~\ref{sec:5}, we investigate the Bell non-locality~\cite{bell1964einstein,brunner2014bell} of $N$-uniform qudit hypergraph states and expound a general detection scheme for illustrating the Bell non-locality of a general qudit hypergraph state. Conclusions are drawn in Sec.~\ref{sec:6}.

\section{multi-hypergraphs and qudit hypergraph states}
\label{sec:2}
In this section, we will introduce some preliminary knowledge of hypergraph and qubit hypergraph state and propose our main generalization of these concepts. Some important properties of qubit hypergraph states and qudit hypergraph states will be discussed.

\subsection{Preliminary: hypergraphs and qubit hypergraph states}
\label{2A}
A hypergraph $H$ is composed of a set of vertices $V$ and a set of hyperedges $E$ \cite{Qu13,Rossi13,Gachechiladze16}, i.e., $H=(V,E)$ (For simplification, in this subsection, $H$ represents such a hypergraph.). Suppose that the vertices are labeled as $1,2,\cdots,N$, then $V=\{1, 2, \cdots, N\}$. Unlike the edges defined in standard graphs, hyperedges in hypergraphs may connect more (or less) than 2 vertices, i.e., elements in $E$ has a form $e=\{k_1, k_2, \cdots, k_{|e|}\}$, where $k_1, k_2, \cdots, k_{|e|}$ are the vertices connected by $e$, and $|e|$, the cardinality of $e$, ranges from 0 to $N$. If all the hyperedges
in $H$ are of the same cardinality $k$, then $H$ is called $k$-uniform \cite{Rossi13}. Standard graphs are in fact 2-uniform hypergraphs. Some examples of hypergraphs are shown in Fig. \ref{fig:1:subfig}.

Hypergraphs can be encoded into a class of quantum states named qubit hypergraph states, in which every vertex represents a two-level quantum system whose computational basis is $\{|0\rangle, |1\rangle\}$. The operator corresponding to the hyperedge $e=\{k_1, \cdots, k_{|e|}\}$ is defined as
\begin{equation}\label{eqn:Ce}
C_e=\begin{cases}
\begin{array}[t]{cc}
-1 & |e|=0,\\
Z  & |e|=1,\\
\underset{i_{k_1},\cdots,i_{k_{|e|}}=0}{\overset{1}{\sum }} {(-1)}^{i_{k_1}\cdots i_{k_{|e|}}} \hat{\Pi}_{i_{k_1} \cdots i_{k_{|e|}}}  & |e|\geq 2,
\end{array}\end{cases}
\end{equation}
where $\hat{\Pi}_{i_{k_1} \cdots i_{k_{|e|}}}=|i_{k_1}\cdots i_{k_{|e|}}\rangle \langle i_{k_1}\cdots i_{k_{|e|}}|$ and $i_{k_1}$, $\cdots$, $i_{k_{|e|}}$ denotes the value of the vertices $k_1, \cdots, k_{|e|}$, respectively. The qubit hypergraph state corresponding to $H$ is
\begin{equation}\label{eqn:H}
|H\rangle=\underset{e\in E}{\prod}C_e{\ket{+}}^{\otimes N},
\end{equation}
where $|+\rangle=(|0\rangle+|1\rangle)/\sqrt{2}$. The state $|H\rangle$ can be interpreted as applying a series of $C_e$ operations to $|+\rangle ^{\otimes N}$. As all the $C_e$s are commutative with respect to each other, the order of the operations makes no difference, and a hypergraph corresponds to a definite qubit hypergraph state (see Fig. \ref{fig:1:subfig} for the examples).

The same as the graph states, qubit hypergraph states can also be characterized within the framework of stabilizers. Define a set of operators
\begin{equation}\label{equation3}
g_k=\Big(\prod _{e\in E} C_e\Big) X_k \Big(\prod _{{e'}\in E} C_{e'}\Big)^{\dagger}=X_k\prod _{\{e|k\in e, e\in E\}} C_{e\backslash \{k\}},
\end{equation}
where $X_k$ is the Pauli-$X$ operator of the $k$th vertex, then (see Fig.~\ref{fig:1:subfig} for the examples)
\begin{equation}\label{equation4}
g_k|H\rangle=|H\rangle.
\end{equation}
Because
\begin{equation}
[g_k,g_{k'}]=\Big(\prod _{e\in E} C_e\Big) [X_k,X_{k'}] \Big(\prod _{{e'}\in E} C_{e'}\Big)^{\dagger }=0,
\end{equation}
the set $\{g_k |k \in V\}$ can generate an Abelian cyclic group called the stabilizer group of $|H\rangle$. Either $\{g_k|k\in V\}$ or the stabilizer group can determine a qubit hypergraph state up to a phase factor~\cite{Hein06,Rossi13}.

\begin{figure}[t]
\includegraphics[width =8.5 cm]{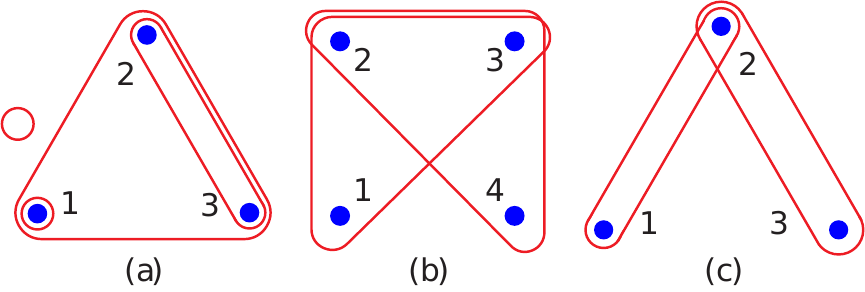}
\caption{Examples of hypergraphs. (a) A common hypergraph with $E=\{\emptyset,\{1\},\{2,3\},\{1,2,3\}\}$ (The red circle in the left represents an empty hyperedge). The corresponding qubit hypergraph state is $|\psi\rangle=(-|000\rangle-|001\rangle-|010\rangle+|011\rangle+|100\rangle+|101\rangle+|110\rangle+|111\rangle)/\sqrt{8}$, whose stabilizer group is generated by $X_1 C_{\emptyset} C_{\{2,3\}}$, $X_2 C_{\{3\}} C_{\{1,3\}}$ and $X_3 C_{\{2\}} C_{\{1,2\}}$. (b) Hypergraph with $E=\{\{1,2,3\},\{2,3,4\}\}$. As the cardinalities of the hyperedges are both 3, this hypergraph is a 3-uniform hypergraph, a generalization of conventional graph. (c) A hypergpraph with $E=\{\{1,2\},\{2,3\}\}$. Because the hyperedges are both of cardinality 2, this hypergraph reduces to a conventional graph and the corresponding qubit hypergraph state becomes a conventional graph state.} \label{fig:1:subfig}
\end{figure}

Qubit hypergraph states have interesting properties and important applications. The formalism offers a systematically pictorial representation of the \emph{real equally weighted state}s, which is a vivid way of demonstrating entanglement~\cite{Rossi13}. The entanglement and Bell non-locality make this class of quantum states have a broad range of applications in quantum computation and quantum metrology~\cite{Guhne14, Gachechiladze16}.

\subsection{Multi-hypergraphs and qudit hypergraph states}
\label{2B}

Multi-hypergraph, whose hyperedge can have a multiplicity larger than 1, is a generalization of hypergraph (see Fig. \ref{fig:2:subfig} for the examples). A multi-hypergraph whose vertices represent $d$-level quantum systems can be denoted as $H_d=\left(V, E\right)$, where $V=\left\{1, 2, \cdots, N\right\}$ is the set of vertices, and $E$ is a multiset of the hyperedges. The times an element $e$ occurs in $E$ is called multiplicity of $e$ and is denoted as $m_e$ $(m_e \in \{1,2,\cdots,d-1\})$~\cite{Looi08,Ionicioiu12}. For the $e$ that satisfies $e\in 2^V$ ($2^V$ denotes the power set of $V$, which constitutes of all the subsets of $V$) and $e\notin E$, its multiplicity $m_e$ is defined to be 0. With this generalization, every $H_d$ is associated with a definite multiplicity function $e \rightarrow m_e$, here $e\in 2^V$ and $m_e\in\{0,1,\cdots,d-1\}$. In the following, if not particularly specified, $H_d$ refers to such a multi-hypergraph, and the multiplicity of $e$ is denoted as $m_e$.

\begin{figure}
\includegraphics[width =8.5 cm]{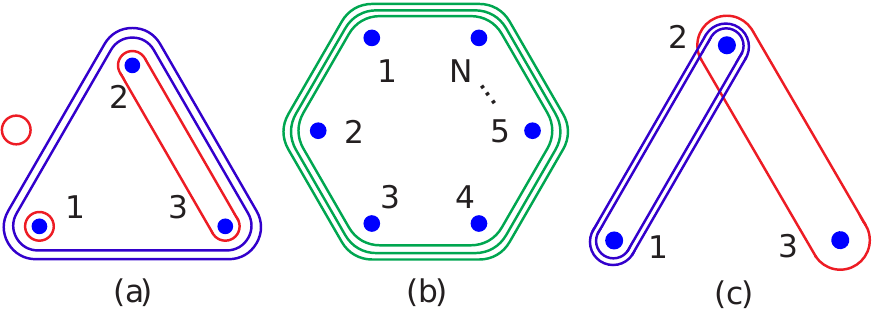}
\caption{Examples of multi-hypergraphs. Hyperedges with different multiplicities are drawn in different colors. (a) A common multi-hypergraph with $E=\{\emptyset,\{1\},\{2,3\},\{1,2,3\},\{1,2,3\}\}$, i.e., $m_{\emptyset}=1$, $m_{\{1\}}=1$, $m_{\{2,3\}}=1$, $m_{\{1,2,3\}}=2$, otherwise $m_e=0$. (The red circle in the left represents an empty hyperedge.) Suppose each vertex represents a qutrit, then the corresponding quantum state is $|H_3\rangle=\sum_{i_1,i_2,i_3=0}^{2} \omega_3^{1+i_1+i_1 i_2+2 i_1 i_2 i_3}|i_1 i_2 i_3\rangle /\sqrt{27}$, where $\omega_3=e^{ i 2 \pi/3}$. The stabilizer group is generated by $X_1 C_{\emptyset}^{\dagger} (C_{\{2,3\}}^{\dagger})^2$, $X_2 C_{\{3\}}^{\dagger} (C_{\{1,3\}}^{\dagger})^2$, and $X_3 C_{\{2\}}^{\dagger} (C_{\{1,2\}}^{\dagger})^2$. (b) An $N$-vertex multi-hypergraph with $m_{\{1,2,\cdots,N\}}=3$ (otherwise, $m_e=0$). This multi-hypergraph is symmetric in permutation of vertices. When encoding this multi-hypergraph into a quantum state, each vertex represents a quantum system whose dimension is larger than 3. (c) A multi-hypergraph with $E=\{\{1,2\},\{1,2\},\{2,3\}\}$. Because all the hyperedges in $E$ are of cardinality 2, this multi-hypergraph is in fact a conventional multi-graph that can be encoded into qudit graph states.}
\label{fig:2:subfig}
\end{figure}

Now we define qudit hypergraph states corresponding to $H_d$. Suppose the computational basis of each vertex is $\{|0\rangle, |1\rangle, \cdots, |d-1\rangle\}$, then in this basis the generalized Pauli-$X$ and Pauli-$Z$ operators are~\cite{Zhou03,Hall07,Sean06,Looi08}
\begin{equation}
\label{eq:XZ}
X=\left[
\begin{array}{ccccc}
 0 & 1 & 0 & \cdots  & 0 \\
 0 & 0 & 1 & \cdots  & 0 \\
 \vdots  & \vdots  & \vdots  & \ddots & \vdots  \\
 0 & 0 & 0 & \cdots  & 1 \\
 1 & 0 & 0 & \cdots  & 0
\end{array}
\right],\\
Z=\left[
\begin{array}{ccccc}
 1 & 0 & 0 & \cdots  & 0 \\
 0 & \omega_d  & 0 & \cdots  & 0 \\
 0 & 0 & \omega ^2_d & \cdots  & 0 \\
 \vdots  & \vdots  & \vdots  & \ddots & \vdots  \\
 0 & 0 & 0 & \cdots  & \omega ^{d-1}_d
\end{array}
\right],
\end{equation}
in which $\omega_d=e^{{i  2 \pi}/{d}}$ and $X Z=\omega_d Z X$ (Manin's quantum plane algebra \cite{ge1992cyclic}). The operator corresponding to the hyperedge $e=\{k_1, k_2, \cdots, k_{|e|}\}$ is defined as
\begin{equation}\label{eq2}
C_e=\begin{cases}
\begin{array}[t]{ccc}
\omega_d & |e|=0,\\
Z  & |e|=1,\\
\underset{{i_{k_1},\cdots, i_{k_{|e|}}}=0}{\overset{d-1}{\sum }}\omega ^{{i_{k_1}\cdots i_{k_{|e|}}}}_d \hat{\Pi}_{i_{k_1} \cdots i_{k_{|e|}}}  & |e|\geq 2,
\end{array}\end{cases}
\end{equation}
where $\hat{\Pi}_{i_{k_1} \cdots i_{k_{|e|}}}=|i_{k_1}\cdots i_{k_{|e|}}\rangle \langle i_{k_1}\cdots i_{k_{|e|}}|$ and $i_{k_1}$, $\cdots$, $i_{k_{|e|}}$ denote the possible values of the vertices $k_1$, $\cdots$, $k_{|e|}$ (in the computational basis), respectively. The unitary operators $X$, $Z$, and $C_e$ satisfy
\begin{equation}\label{eq:Iden}
\begin{split}
X^k&=\mathbb{I} \iff k=0 \pmod d,\\
Z^k&=\mathbb{I} \iff k=0 \pmod d,\\
C_e^k&=\mathbb{I} \iff k=0 \pmod d.\\
\end{split}
\end{equation}
Denote that $\left|+\right\rangle_d=\sum_{k=0}^{d-1}{{|k\rangle}/{\sqrt{d}}}$,
then the $d$-level hypergraph state corresponding to $H_d$ can be defined as
\begin{equation}\label{eq:def}
|H_d\rangle =\underset{e\in {2^V}}{\prod }C_e^{m_e}{\left|+\right\rangle}^{\otimes N}_d.
\end{equation}
Here the condition ``$e\in {2^V}$'' is equivalent to ``$e\in E$'', because $C_e^0=\mathbb{I}$ $(\forall e\in{2^V})$. For simplicity, in the following we will not express it explicitly.

A qudit hypergraph state is also associated with a stabilizer group through which it can be determined up to a phase factor. For $H_d=\left(V,E\right)$, define
\begin{equation}\label{eq11}
g_k=\l(\prod C_e^{m_e}\r) X_k \l(\prod C_{e'}^{m_{e'}}\r)^{\dagger}=X_k\underset{e:k\in e}{\prod }\left(C_{e\backslash \{k\}}^{\dagger }\right){}^{m_e},
\end{equation}
then
\begin{equation}\label{eq10}
\begin{split}
g_k \left.\left|H_d\right.\right\rangle=\left.\left|H_d\right.\right\rangle,
\end{split}
\end{equation}
and
\begin{equation}
[g_k,g_{k'}]=\l(\prod C_e^{m_e}\r) [X_k,X_{k'}] \l(\prod C_{e'}^{m_{e'}}\r)^{\dagger}=0.
\end{equation}
 
Note that the form of $g_k$ in Eq. \eqref{eq11} is different from that in Eq. \eqref{equation3}. The reason is that when $d=2$, $\forall e$, 
$C_e$ is Hermitian, while for general $d$, this property cannot always hold. 
\black 
The set $\{g_k|k\in V\}$ generates a cyclic Abelian group named the stabilizer group of $\left.\left|H_d\right.\right\rangle$. Generally speaking, like those of qubit hypergraph states~\cite{Guhne14}, the stabilizers of qudit hypergraph states are also non-local operators.

\section{relation between multi-hypergraphs and qudit hypergraph states: correspondence and entanglement property}
\label{sec:3}
In this section, we will discuss the relation between multi-hypergraphs and qudit hypergraph states. Theorem \ref{theorem1} shows that the map from $\{H_d|H_d=(V,E)\}$ to $\{\ket{H_d}|H_d=(V,E)\}$, where $H_d$ is mapped to $\ket{H_d}$, is a bijection. Theorem \ref{theorem2} demonstrates that the connectivity of a multi-hypergraph is closely related to the entanglement property of the corresponding quantum state. To prove these two theorems, we shall prove several lemmas first.

\newtheorem{lemma}{Lemma}
\begin{lemma}\label{lem1}
Divide the hyperedge $e=\{1, 2,\cdots, n\}$ into the control part $e_C=\{1, 2, \cdots, m\}$ and the target part $e_T=\{m+1, m+2, \cdots, n\}$, then
\begin{equation}\label{lem12}
C_e=\underset{i_1,\cdots ,i_{m}=0}{\overset{d-1}{\sum }}|i_1 \cdots i_{m}\rangle \langle i_1 \cdots i_{m}|C_{e_T}^{i_1 \cdots i_m}.
\end{equation}
\end{lemma}

\begin{proof}
From the definition in Eq. \eqref{eq2},
\begin{equation}
\begin{split}
C_e&=\underset{i_{1},\cdots ,i_n=0}{\overset{d-1}{\sum }}\omega _d ^{i_{1} \cdots i_n}\hat{\Pi}_{i_1 \cdots i_n},\\
C_{e_C}&=\underset{i_{1},\cdots ,i_m=0}{\overset{d-1}{\sum }}\omega _d ^{i_{1} \cdots i_m}\hat{\Pi}_{i_1 \cdots i_m},\\
C_{e_T}&=\underset{i_{m+1},\cdots ,i_n=0}{\overset{d-1}{\sum }}\omega _d^{i_{m+1} \cdots i_n}\hat{\Pi}_{i_{m+1} \cdots i_n}.\\
\end{split}
\end{equation}
Because
\begin{equation}
\begin{split}
&\underset{i_{m+1},\cdots ,i_n=0}{\overset{d-1}{\sum }}\omega_d ^{i_1 \cdots i_n} \hat{\Pi}_{i_1 \cdots i_n}\\
=&\underset{i_{m+1},\cdots ,i_n=0}{\overset{d-1}{\sum }}\omega_d ^{i_1 \cdots i_n} \hat{\Pi}_{i_1 \cdots i_m} \hat{\Pi}_{i_{m+1} \cdots i_n}\\
=&\hat{\Pi}_{i_1 \cdots i_m} \underset{i_{m+1},\cdots ,i_n=0}{\overset{d-1}{\sum }}\left(\omega_d^{i_{m+1} \cdots i_{n}}\right)^{i_1 \cdots i_m}\hat{\Pi}_{i_{m+1} \cdots i_n}\\
=&\hat{\Pi}_{i_1 \cdots i_m} C_{e_T}^{i_1 \cdots  i_m},
\end{split}
\end{equation}
\begin{equation}
\begin{split}
C_e&=\underset{i_1, \cdots, i_n=0}{\overset{d-1}{\sum }}\omega_d ^{i_1 \cdots i_n} \hat{\Pi}_{i_1 \cdots i_n}\\
&=\underset{i_1, \cdots, i_m=0}{\overset{d-1}{\sum }}\underset{i_{m+1}, \cdots, i_n=0}{\overset{d-1}{\sum }}\omega_d ^{i_1 \cdots i_n}\hat{\Pi}_{i_1 \cdots i_n}\\
&=\underset{i_1, \cdots, i_m=0}{\overset{d-1}{\sum}}\hat{\Pi}_{i_1 \cdots i_m} C_{e_T}^{i_1 \cdots i_m},
\end{split}
\end{equation}
which is exactly the conclusion in Lemma~\ref{lem1}.
\end{proof}

Lemma~\ref{lem1} demonstrates that a hyperedge operation can be interpreted as a controlled operation: the products of the vertices in $C$ determine the operations imposed on the target part $T$. In fact, one can choose an arbitrary subset of $e$ as the control part, and the remaining part as the target, which originates from the symmetry of $C_e$. \black

\begin{lemma}\label{lemma:2}
Consider a system composed of $A$ and $B$, whose associated Hilbert spaces are $\mathcal{H}_A$ and $\mathcal{H}_B$, respectively. Suppose $\{|1\rangle, |2\rangle, \cdots, |n\rangle\}$ is an orthonormal basis of $\mathcal{H}_A$ and $|\psi _1\rangle$, $|\psi _2\rangle$, $\cdots$, $|\psi _n\rangle$ are normalized vectors in $\mathcal{H}_B$. The vector
\begin{equation}
\left.\left.\left.\left.\left.\left.\left.|1\rangle\left|\psi_1\right.\right\rangle +\right|2\right\rangle |\psi _2\right\rangle +\cdots +\right|n\right\rangle |\psi _n\right\rangle
\end{equation}
is a product state if and only if all the $|\psi_j\rangle$s $(1\leq j\leq n)$ are parallel.
\end{lemma}

\begin{proof}
 (i)``if''. If all the $|\psi _j\rangle$s are parallel, then each $|\psi _j\rangle$ has a form $\left.e^{i \phi_j}|\psi _0\right\rangle$. So
\begin{equation}
|1\rangle|\psi _1\rangle +|2\rangle |\psi_2\rangle +\cdots +|n\rangle |\psi _n\rangle =\Big(\sum _{j=1}^n e^{i \phi_j}|j\rangle\Big) |\psi _0\rangle,
\end{equation}
which is a product state.

(ii)``only if". Suppose the total system is in a product state, $B$ remains the same physical state no matter what measurement is made to $A$ and whatever the result is. By implementing the von Neumann measurement $\{M_j=|j \rangle \langle j||j\in\{1, 2, \cdots, n\}\}$ to $A$, the part $B$ will collapse to one of the states in $\l\{|\psi_1\rangle, |\psi_2\rangle, \cdots, |\psi_n\rangle\r\}$. So all the $|\psi_j\rangle$s are physically equivalent, i.e., they are parallel.
\end{proof}

\begin{lemma}\label{lemma:ordering}
Qudit hypergraph state $|H_d\rangle$ equals $|+\rangle ^{\otimes N}_d$ if and only if $E=\emptyset$.
\end{lemma}
\begin{proof}
(i)``if''. If $E=\emptyset$, by definition for all $e\in{2^V}$, $m_e=0$, so $|H_d\rangle=|+\rangle ^{\otimes N}_d$.

(ii)``only if". The stabilizer group of $\left.\left|H_d\right.\right\rangle$ is generated by $\left\{X_k\underset{e:k\in e}{\prod }\left(C_{e\backslash \{k\}}^{\dagger }\right){}^{m_e} \middle| k\in V\right\}$ while that of $\ket{+}^{\otimes N}_d$ is generated by $\left\{X_k |k\in V\right\}$.

If
\begin{equation}\label{lem30}
\underset{e}{\prod }C_e^{m_e}{\left|+\right\rangle}^{\otimes N}_d=|+\rangle ^{\otimes N}_d,
\end{equation}
the two qudit hypergraph states will have the same stabilizer group, leading to
\begin{equation}\label{lem3}
X_k \underset{e:k\in e}{\prod }\left(C_{e\backslash \{k\}}^{\dagger }\right){}^{m_e}=X_k \underset{j \neq k}{\prod } X^{p_j}_j,
\end{equation}
where $k \in V$ and $p_j \in \{0,1,\cdots,d-1\}$. The factor $\underset{e:k\in e}{\prod }\left(C_{e\backslash \{k\}}^{\dagger }\right){}^{m_e}$
is always diagonal in the computational basis, while $\underset{j \neq k}{\prod } X^{p_j}_j$ is diagonal only if $p_j=0$ ($\forall j \neq k$). That is to say, to make Eq. \eqref{lem3} hold,
\begin{equation}
\underset{e:k\in e}{\prod }\left(C_{e\backslash \{k\}}^{\dagger }\right){}^{m_e}=\underset{j \neq k}{\prod } X^0_j=\mathbb{I},
\end{equation}
thus (notice that $C_{e\backslash \{k\}}$ is unitary)
\begin{equation}
\underset{e:k\in e}{\prod }C_{e\backslash \{k\}}^{m_e}|+\rangle{}^{\otimes N-1}_d=|+\rangle {}^{\otimes N-1}_d.
\end{equation}
Implement the above procedure several times, generally, one arrives
\begin{equation}\label{lem32}
\underset{e:k_1,\cdots,k_n\in e}{\prod }C_{e\backslash \{k_1,\cdots,k_n\}}^{m_e}|+\rangle {}^{\otimes N-n}_d=|+\rangle {}^{\otimes N-n}_d.
\end{equation}

When $n=N-1$, Eq. \eqref{lem32} becomes
\begin{equation}
C_{\emptyset}^{m_{\{k_1,k_2,\cdots,k_{N-1}\}}}C_{\{k_N\}}^{m_{\{k_1,k_2,\cdots,k_{N}\}}}|+\rangle_d=|+\rangle_d,
\end{equation}
indicating that
\begin{equation}
m_{\{k_1,k_2,\cdots,k_{N-1}\}}=m_{\{k_1,k_2,\cdots,k_{N}\}}=0.
\end{equation}
Because all the $k_i$ $(i \in \{1,2,\cdots,N\})$ are arbitrarily arranged in order, for all the $e$ that satisfy $|e|=N$ or $N-1$, $m_{e}=0$.

When $n=N-2$, Eq. \eqref{lem32} becomes
\begin{equation}\label{lem33}
\underset{e:k_1,\cdots,k_{N-2}\in e}{\prod }C_{e\backslash \{k_1,\cdots,k_{N-2}\}}^{m_e}|+\rangle {}^{\otimes 2}_d=|+\rangle {}^{\otimes 2}_d.
\end{equation}
The product involves all the hyperedges containing $\{k_1,\cdots,k_{N-2}\}$, i.e., the cardinalities of these hyperedges are larger or equal to $N-2$. As is shown in the previous paragraph, hyperedges whose cardinalities are larger than $N-2$ must have 0 multiplicity, thus contributing to identity factors. So Eq. \eqref{lem33} can be reduced to
\begin{equation}
C_{\emptyset}^{m_{\{k_1,k_2,\cdots,k_{N-2}\}}}|+\rangle{}^{\otimes 2}_d=|+\rangle{}^{\otimes 2}_d,
\end{equation}
indicating that ${m_{\{k_1,k_2,\cdots,k_{N-2}\}}}=0$. Generally, if $|e|=N-2$, $m_e=0$.

Similarly, for all the $e$ that satisfy $|e|=N-3,N-2,\cdots,0$, $m_e=0$.

So if $|H_d\rangle=|+\rangle ^{\otimes N}_d$, $m_e=0$ $(\forall e\in 2^V)$, i.e., $E=\emptyset$.
\end{proof}

With these lemmas, we can prove the following theorems.
\newtheorem{theorem}{Theorem}
\begin{theorem}
\label{theorem1}
Suppose $H_{d}'=\left(V, E'\right)$ and $H_{d}=\left(V, E\right)$, then $\left|H_d'\right\rangle = \left|H_{d}\right\rangle$ if and only if $E'=E$.
\end{theorem}
\begin{proof}
(i)``if''. By definition, in terms of representing $d$-level hypergraph states, a multi-hypergraph corresponds to a unique $d$-level hypergraph state.

(ii)``only if''. For $e \in {2^V}$, denote its multiplicity corresponding to $H_{d}'$ as $m_e'$, then $|H'_d\rangle=\underset{e}{\prod }C_e^{m'_e}{\left|+\right\rangle}^{\otimes N}_d$.
If $\left|H_d'\right\rangle = \left|H_{d}\right\rangle$,
\begin{equation}
{\left|+\right\rangle}^{\otimes N}_d=\underset{e}{\prod }C_e^{m'_e-m_e}{\left|+\right\rangle}^{\otimes N}_d.
\end{equation}
According to Lemma~\ref{lemma:ordering}, this equation holds if and only if for all $e$, $m'_e-m_e=0$, i.e., $E'=E$.
\end{proof}

Theorem \ref{theorem1} indicates that distinct multi-hypergraphs correspond to distinct quantum states, assuming that the systems are both $N$-qudit systems.

An important entanglement property of qudit hypergraph states is revealed in the following theorem.

\begin{theorem}
\label{theorem2}
If one part of a multi-hypergraph is connected with the other part, then these two corresponding subsystems are entangled.
\end{theorem}
\begin{proof}
Suppose $H_d=(V,E)$, divide $V$ into two parts, one is called the control part ($C=\{c_1,c_2,\cdots,c_{|C|}\}$) and the other is called the target ($T=\{t_1,t_2,\cdots,t_{|T|}\}$), satisfying $C\cup T=V$ and $C\cap T=\emptyset$. \black Accordingly, we can define 3 sub-multisets of $E$, i.e., $E_C$, $E_T$ and $\Lambda$. $E_C$ ($E_T$) constitutes of all the elements in $E$ that are subsets of $C$ ($T$); $\Lambda$ consists of all the elements in $E$ that contains vertices in $C$ and $T$ simultaneously. If $\Lambda \neq \emptyset$, $C$ and $T$ are connected through hyperedges in $\Lambda$. \black

Define the multi-hypergraphs $H_d^C=(C,E_C)$ and $H_d^T=(T,E_T)$, then
\begin{equation}\label{eq:theorem1}
|H_d\rangle =C_{\emptyset}^{-m_\emptyset} \underset{\epsilon \in \Lambda}{\prod} C_{\epsilon}^{m_{\epsilon}} |H_d^{C}\rangle |H_d^{T}\rangle,
\end{equation}
where $|H_d^{C}\rangle=\prod_{e'\in E_C} C_{e'}^{m_{e'}}{\left|+\right\rangle}^{\otimes |C|}_d$ and $|H_d^{T}\rangle=\prod_{e''\in E_T} C_{e''}^{m_{e''}}{\left|+\right\rangle}^{\otimes |T|}_d$ (notice that the multiplicity of each hyperedge in $E_C$, $E_T$ and $\Lambda$ is the same as the one in $E$).
Expanding $|H_d^{C}\rangle$ in the computational basis explicitly, one has
\begin{equation}
|H_d^{C}\rangle=\frac{1}{{\sqrt{d}}^{|C|}}\underset{i_{c_1},\cdots,i_{c_{|C|}}=0}{\overset{d-1}{\sum}}e^{i \phi(i_{c_1},\cdots ,i_{c_{|C|}})}|i_{c_1}  \cdots i_{c_{|C|}}\rangle.
\end{equation}
According to Lemma~\ref{lem1}, all the $C_e$s in Eq. \eqref{eq:theorem1} can be expressed in a form like Eq. \eqref{lem12}, so
\begin{equation}\label{eqtheorem}
|H_d\rangle=\frac{C_{\emptyset}^{m_\emptyset}}{{\sqrt{d}}^{|C|}}\underset{i_{c_1},\cdots,i_{c_{|C|}}=0}{\overset{d-1}{\sum}}|i_{c_1}\cdots i_{c_{|C|}}\rangle' \hat{f}(i_{c_1},\cdots ,i_{c_{|C|}})|H_d^T\rangle,
\end{equation}
where $|i_{c_1}\cdots i_{c_{|C|}}\rangle'=e^{i \phi(i_{c_1},\cdots ,i_{c_{|C|}})}|i_{c_1}  \cdots i_{c_{|C|}}\rangle$ and $\hat{f}(i_{c_1},\cdots ,i_{c_{|C|}})$ is some composite hyperedge transformation.

If $|H_d\rangle$ is a product state, all $\hat{f}(i_{c_1},i_{c_2},\cdots,i_{c_{|C|}})|H_d^T\rangle$ ($\forall i_{c_1},\cdots,i_{c_{|C|}}\in\{0,1,\cdots,d-1\}$) must be parallel (Lemma~\ref{lemma:2}), i.e.,
\begin{equation}
\begin{split}
\hat{f}(i_{c_1},\cdots ,i_{c_{|C|}})|H_d^T\rangle&=e^{i \delta(i_{c_1},\cdots ,i_{c_{|C|}})}\hat{f}(0,\cdots ,0)|H_d^T\rangle\\
&=e^{i \delta(i_{c_1},\cdots ,i_{c_{|C|}})}|H_d^T\rangle.
\end{split}
\end{equation}
Divide every $\epsilon$ in $\Lambda$ into $c_{\epsilon}$ and $t_{\epsilon}$, where $c_{\epsilon}=\epsilon \cap C$ and $t_{\epsilon}=\epsilon \cap T$, then (Lemma \ref{lem1}),
\begin{equation}
\hat{f}(1,\cdots,1)|H_d^T\rangle=\underset{\epsilon \in {\Lambda}}{\prod} C_{t_{\epsilon}}^{m_{\epsilon}}|H_d^T\rangle.
\end{equation}
So
\begin{equation}
\underset{\epsilon \in {\Lambda}}{\prod} C_{t_{\epsilon}}^{m_{\epsilon}}|H_d^T\rangle=e^{i \delta(1,\cdots ,1)}|H_d^T\rangle=C_{\emptyset}^z|H_d^T\rangle,
\end{equation}
where $z\in \{0,1,\cdots,d-1\}$, thus
\begin{equation}
C_{\emptyset}^{d-z}\underset{\epsilon \in {\Lambda}}{\prod} C_{t_{\epsilon}}^{m_{\epsilon}}{\left|+\right\rangle}^{\otimes |T|}_d={\left|+\right\rangle}^{\otimes |T|}_d.
\end{equation}
This equation cannot be true because of Lemma \ref{lemma:ordering}. So $|H_d\rangle$ cannot be a product state in the form like $|\psi \rangle _C|\phi \rangle _T$, i.e., the two parts are entangled.
\end{proof}

Theorem \ref{theorem2} offers us an ability to knowing the entanglement structure of a qudit hypergraph state by reading the connectivity property of the multi-hypergraph. With the result in this theorem, we have the following two corollaries.

\newtheorem{corollary}{Corollary}
\begin{corollary}\label{coro1}
If a multi-hypergraph $H_d$ is connected, then $|H_d\rangle$ is genuinely entangled.
\end{corollary}
\begin{proof}
If $H_d$ is connected, divide it into arbitrary two parts, then the two parts are connected through some hyperedges. According to Theorem \ref{theorem2}, these two parts are entangled. As the division is arbitrary, $|H_d\rangle$ is non-biseparable, i.e., it is genuinely entangled.
\end{proof}

\begin{corollary}\label{coro2}
Suppose an unconnected multi-hypergraph $H_d$ is composed of several blocks $(H_d^{(i)})$ that are not connected to each other, and each one is a connected multi-hypergraph or possesses only one vertex, then each $\ket{H_d^{(i)}}$ that possesses more than one vertex is a genuinely entangled state, and different blocks are not entangled with each other.
\end{corollary}
\begin{proof}
Different blocks are not connected to each other, so they are not entangled (Eq. \eqref{eq:def}). For connected $H_d^{(i)}$, because $\ket{H_d^{(i)}}$ is also a qudit hypergraph state, it is genuinely entangled (Corollary \ref{coro1}).
\end{proof}

Corollary \ref{coro1} and Corollary \ref{coro2} enable multi-hypergraph a useful tool for visualizing the entanglement of its corresponding qudit hypergraph state.

\section{Relationship among qudit hypergraph states and some well-known state classes}
\label{sec:?}
In this section, we will discuss relationships among qudit hypergraph states and some well-known state classes, i.e., generalized real equally weighted states, qudit graph states, and stabilizer states. 

\subsection{Qudit hypergraph states and generalized real equally weighted states}
\label{}
The \emph{real equally weighted state}s are the quantum states that all the coefficients in the computational basis are real and with equal absolute value. For example, \emph{real equally weighted state}s describing $N$-qubit systems can all be represented in the form 
\begin{equation}
\ket{\psi(f,N)}=\frac{1}{2^{N/2}}\sum_{i_1,\cdots,i_N=0}^1 (-1)^{f(i_1,\cdots,i_N)}\ket{i_1 \cdots i_N}, 
\end{equation}
where $f(i_1,\cdots,i_N)\in \mathbb{Z}_2$. By interpretating $-1$ as $\omega_2$, the generalized real equally weighted states (\emph{GREWS}s) can be expressed as
\begin{equation}\label{eq:generews}
\ket{\psi(f,N)}_d=\frac{1}{d^{N/2}}\sum_{i_1,\cdots,i_N=0}^{d-1} {\omega}_d^{f(i_1,\cdots,i_N)}\ket{i_1 \cdots i_N}, 
\end{equation}
in which $f(i_1,\cdots,i_N)\in \mathbb{Z}_d$.

It has been demonstrated in the literature that qubit hypergraph states are equivalent to \emph{real equally weighted state}s \cite{Qu13, Rossi13}. For the qudit case, it would be interesting to investigate whether a similar relationship exists. From the definition of qudit hypergraph states, we can see that every $N$-qudit hypergraph state can be expressed in the form of Eq. \eqref{eq:generews}, i.e., all qudit hypergraph states are \emph{GREWS}s. \black For specific $N$ and $d$, the total number of \emph{GREWS}s is $d^{d^N}$, while in total there are only $d^{2^N}$ qudit hypergraph states (There are $d^{2^N}$ such multi-hypergraphs in total and Theorem \ref{theorem1} shows that the states and multi-hypergraphs have a one-to-one correspondence).  Only if $d=2$, $d^{d^N}=d^{2^N}$, otherwise, $d^{d^N} > d^{2^N}$. This indicates that if $d>2$, the set of qudit hypergraph states is a proper subset of \emph{GREWS}s. This relationship is different from the qubit case (See Fig. \ref{fig:relation2}). \black

%Take 2-qutrit states for example. Label the vertices as 1 and 2, then the possible hyperedges are $\emptyset$, $\{1\}$, $\{2\}$, and $\{1,2\}$, and their corresponding multiplicities $m_{\emptyset}$, $m_{\{1\}}$, $m_{\{2\}}$ and $m_{\{1,2\}}$ range from 0 to 2. These states then can be expressed in the form
%\begin{equation}
%\ket{H_3}=\frac{1}{3}\sum_{i_1,i_2=0}^2 \omega_3^{m_{\emptyset}+m_{\{1\}}i_1+m_{\{2\}}i_2+m_{\{1,2\}}i_1 i_2} \ket{i_1 i_2}.
%\end{equation}
%The total number of these states is $3^4=81$, while that of states in the form of Eq. \eqref{eq:generews} ($N=2, d=3$) is $3^9=19683$, meaning that most \emph{GREWS}s ($N=2$, $d=3$) cannot be represented by multi-hypergraphs. In another perspective, for a given \emph{GREWS} ($N=2$, $d=3$), 4 coefficients in the computation basis are usually sufficient to solve $m_{\emptyset}$, $m_{\{1\}}$, $m_{\{2\}}$ and $m_{\{1,2\}}$, the other five components maybe controvertial to this solution, meaning that this state cannot be represented by a multi-hypergraph.

\begin{figure}
\includegraphics[width =8 cm]{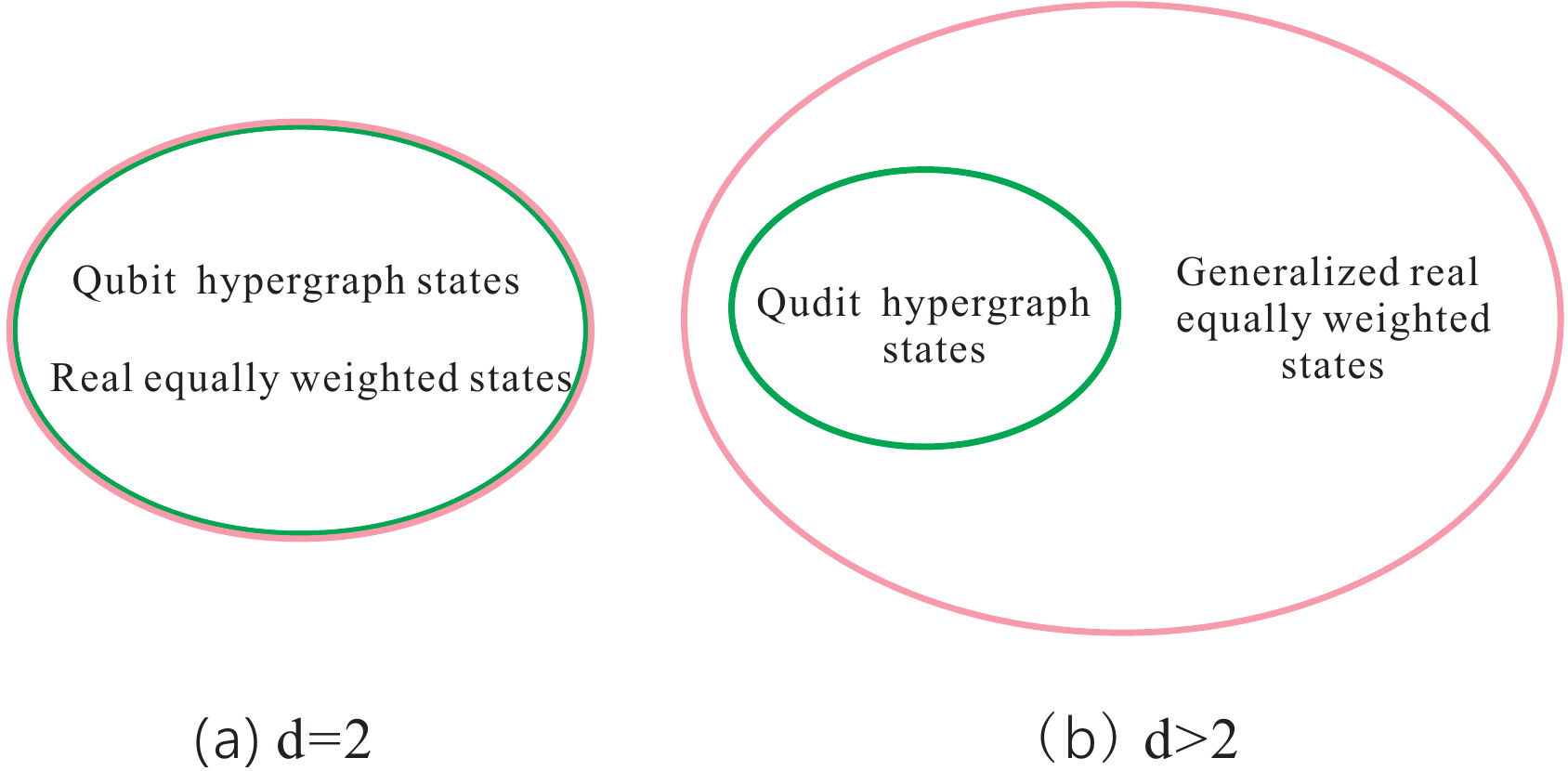}
\caption{Relationship between ``qudit hypergraph states'' and ``generalized real equally weighted states''. (a) When $d=2$, ``qudit hypergraph states'' reduces to qubit hypergraph states, ``generalized real equally weighted states'' reduces to real equally weighted states, and the two sets are equivalent. (b) When $d>2$, qudit hypergraph states form a proper subset of generalized real equally weighted states.}
\label{fig:relation2}
\end{figure}

\subsection{Relationship among qudit hypergraph states, qudit graph states, and stabilizer states}
\label{}
Qudit hypergraph state is a generalization of qudit graph state, so qudit graph states form a subclass of qudit hypergraph states. According to Theorem \ref{theorem1}, two qudit hypergraph states are equal only if their corresponding multi-hypergraphs are the same. Generally speaking, a multi-hypergraph can have hyperedges with cardinalities larger than 2, which is different from that of multigraphs. Therefore, in general, a qudit hypergraph state is not a qudit graph state.

Stabilizer states of $N$-qudit systems are the common eigenstates with eigenvalue 1 of $N$ independent elements in the Pauli group $\mathcal{G}_N^{(d)}$ \cite{ashikhmin2001nonbinary,ketkar2006nonbinary,helwig2013absolutely}, where $\mathcal{G}_N^{(d)}$ is the $N$-fold product of $\mathcal{G}^{(d)}$, and $\mathcal{G}^{(d)}=\l\{\omega_d^a X^b Z^c |a, b, c \in \mathbb{Z}_d\r\}$ ($X$ and $Z$ are the qudit Pauli operators defined by Eq. \eqref{eq:XZ}). 
According to this definition, qudit graph states are all stabilizer states, because there are $N$ independent stabilizers that can be expressed in the form $g_k=X_k \underset{n:\{k,n\}\in E}{\prod }Z_n^{d-m_{\{k,n\}}}$, i.e., $g_k\in \mathcal{G}_N^{(d)}$ ($k\in \{1,2,\cdots,N\} $). As for the relationship between qudit hypergraph states and stabilizer states, we illustrate the result in the following proposition.
\newtheorem{proposition}{Proposition}
\begin{proposition}\label{coro1}
A qudit hypergraph state is a stabilizer state if and only if the cardinalities of the hyperedges are all no more than 2.
\end{proposition}
\begin{proof}
The stabilizer group of $\ket{H_d}$ is generated by $\{g_k=X_k\underset{e:k\in e}{\prod }C_{e\backslash \{k\}}^{d-m_e}|k=1,2,\cdots,N\}$. If the cardinalities of the hyperedges are all no more than 2, then $\forall e, k$, $C_{e\backslash \{k\}}$ is $\omega_d$ or a $Z$ operator. Thus in this case, $\ket{H_d}$ must be a stabilizer state. If some hyperedge in $H_d$ has cardinality larger than 2 (suppose the vertex $k$ is included by such a hyperedge), then $g_k\notin \mathcal{G}_N^{(d)}$. The reason is as follows. If $g_k \in {\mathcal{G}_N}$, then $X_k^{-1} g_k \in {\mathcal{G}_N}$. Define a new qudit hypergraph state that $\ket{H_d(k)}=\underset{e:k\in e}{\prod }C_{e\backslash \{k\}}^{d-m_e} \ket{+}_d^{\otimes N}$, then it must be a product state. If a hyperedge $e$ satisfies $|e|>2$, $H_d(k)$ possesses a hyperedge $e\backslash \{k\}$ satisfying $|e\backslash \{k\}|\geq2$, which means that some vertices in $H_d(k)$ are connected by $e\backslash \{k\}$. According to Theorem \ref{theorem2}, such a qudit hypergraph state cannot be a product state, which is contrary to that $\ket{H_d(k)}$ is a product state. \black So only if the cardinalities of all the hyperedges are no more than 2 can $\ket{H_d}$ be a stabilizer state. \black
\end{proof}

According to Proposition \ref{coro1}, a qudit hypergraph state that is also a stabilizer state at the same time may not be a qudit graph state (See Fig.~\ref{fig:relation1}). It may also be a qudit graph state operated by some generalized local Pauli operations.

To summarize, the relationship among qudit hypergraph states, qudit graph states and stabilizer states can be expressed in Fig. \ref{fig:relation1}, which is very similar to the qubit case studied in Ref. \cite{Qu13}. 

\begin{figure}
\includegraphics[width =7 cm]{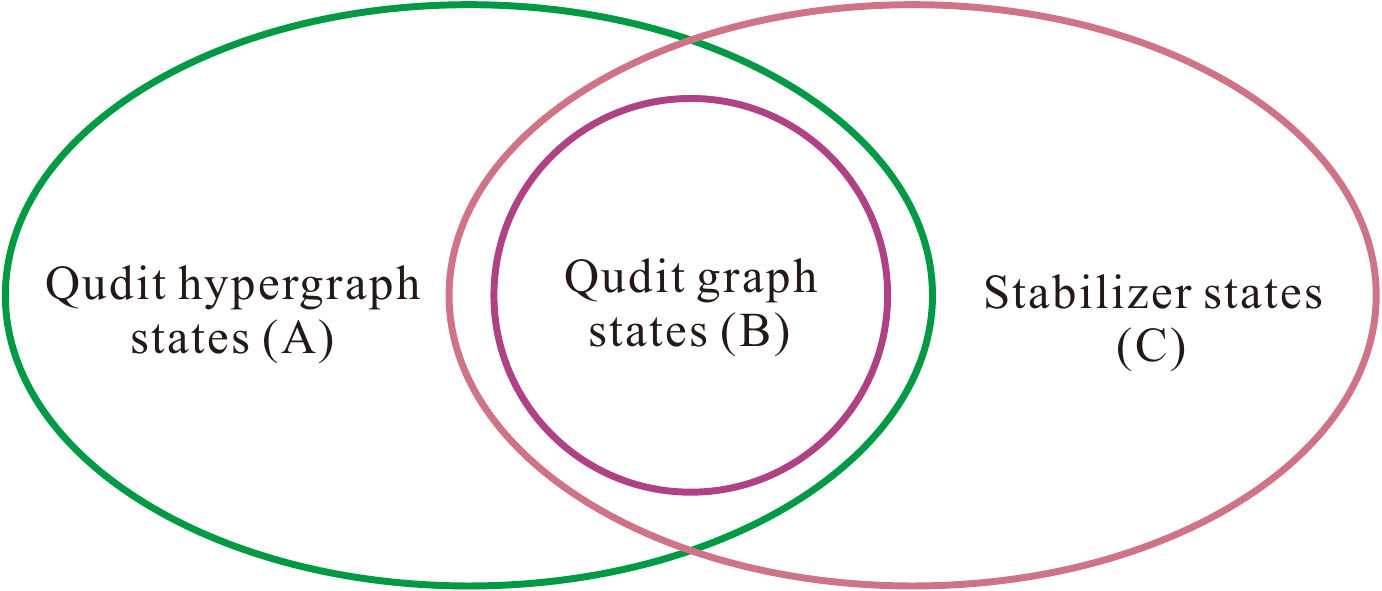}
\caption{Relationship among qudit hypergraph states (A), qudit graph states (B), and stabilizer states (C). B is a proper subset of $A \cap C$, because there are qudit hypergraph states that are qudit graph states acted upon by single-vertex hyperedge operations and 0-vertex hyperedge operations, i.e., they are stabilizer states but not qudit graph states.}
\label{fig:relation1}
\end{figure}

\black

\section{QUDIT HYPERGRAPH STATES' RESPONSES TO local OPERATIONS and measurements}
\label{sec:4}
In this section, we will consider how qudit hypergraph states response to the generalized $Z$, $X$ operations, and generalized $Z$, $X$ measurements~\cite{Qu13}.
The initial states are all assumed to be $\left|H_d\right \rangle$ and the final states are all denoted as $|\psi_f\rangle$.

The local unitary operator $Z_k$, which acts upon the $k$th vertex, can also be interpreted as the hyperedge operation $C_{\{k\}}$. So when $Z_k$ is
implemented on $\left|H_d\right \rangle$,
\begin{equation}
\begin{split}
|\psi_f\rangle &=C_{\{k\}}\left|H_d\right \rangle,
\end{split}
\end{equation}
which is also a qudit hypergraph state. Denote the corresponding multi-hypergraph as $H_d'=(V,E')$, then the associated multiplicity function is
\begin{equation}
m_e'=\begin{cases}
\begin{array}[t]{cc}
m_e & e\neq\{k\},\\
m_{\{k\}}+1 \pmod d & e=\{k\}.
\end{array}\end{cases}
\end{equation}
In pictorial representation, the multiplicity of the hyperedge $\{k\}$ increases by 1 (when $m_{\{k\}}=d-1$, the hyperedge cancels, due to Eq. \eqref{eq:Iden}) while the multiplicities of all the other hyperedges do not change.

When $X_k$ is implemented upon $\left|H_d\right\rangle$, the final state is (see the detailed calculation in Appendix~\ref{A2})
\begin{equation}
\label{eqn:Xk}
\begin{split}
|\psi_f\rangle=\underset{e:k \in e}{\prod }C_{{e}\backslash \{k\}}^{m_{e}}|H_d\rangle,    
\end{split}
\end{equation}
which is also a $d$-level hypergraph state. Denote the corresponding multi-hypergraph as $H_d'=(V,E')$, then the associated multiplicity function is
\begin{equation}
\label{eqn:Xk2}
m_e'=\begin{cases}
\begin{array}[t]{cc}
m_e & k\in e,\\
m_e+m_{e\cup\{k\}} \pmod d & k \notin e.
\end{array}\end{cases}
\end{equation}
Equations~\eqref{eqn:Xk} and \eqref{eqn:Xk2} indicate that the multi-hypergraph corresponding to the final state is the one that the multiplicities of all the hyperedges in the form $e\backslash \{k\}$ are added by $m_{e}$ ($k\in e$) with respect to that of $H_d$ (may also subtract a multiple of $d$ in order to keep $m_e\in \{0,1,\cdots,d-1\}$), while the multiplicities of the other hyperedges do not change.

The cyclic group generated by $\left\{X_k, \left.Z_k\right|k\in V\right\}$ preserves the set structure of qudit hypergraph states. For any element in the group, the action on the qudit hypergraph state can be represented by adding related hyperedges to the hypergraph.

For a general $d$-level system, $Z$ and $X$ operators are unitary but usually non-Hermitian. As $Z$ and $X$ both possess orthonormal eigenvectors, one can still define von Neumann measurements of $Z$ and $X$ by associating each eigenvector with a real value. After the measurement, the measured qudit collapses to an eigenstate of $Z$ ($X$). Correspondingly, the system composed of the remaining qudits collapses to a new state.

Suppose we measure $Z_k$ and the vertex collapses to $|i_k\rangle$ (denote that $\hat{\Pi}_{i_k}=|i_k\rangle \langle i_k|$), then the whole system collapses to
\begin{equation}
\begin{split}
\hat{\Pi}_{i_k}|H_d\rangle=&\Big(\underset{e':k\notin e'}{\prod }C_{e'}^{m_{e'}}\Big)\hat{\Pi}_{i_k}\Big(\underset{e:k\in e}{\prod } C_{e}^{m_{e}}\Big)\left|+\right\rangle {}^{\otimes N}_d.
\end{split}
\end{equation}
While when $k\in e$,
\begin{equation}
\begin{split}
\hat{\Pi}_{i_k}C_e=\hat{\Pi}_{i_k}\underset{j_k=0}{\overset{d-1}{\sum }}\hat{\Pi}_{j_k}C_{e\backslash \{k\}}^{j_k}=C_{e\backslash \{k\}}^{i_k}\hat{\Pi}_{i_k},
\end{split}
\end{equation}
so
\begin{equation}
\begin{split}
\hat{\Pi}_{i_k}\underset{e:k\in e}{\prod }C_e^{m_e}\left|+\right\rangle ^{\otimes N}_d=&\underset{e:k\in e}{\prod }C_{e\backslash \{k\}}^{i_k m_e}\hat{\Pi}_{i_k}\left|+\right\rangle ^{\otimes N}_d \\
=&\frac{1}{\sqrt{d}}|i_k\rangle \underset{e:k\in e}{\prod }C_{e\backslash \{k\}}^{i_k m_e}|+\rangle ^{\otimes N-1}_d.
\end{split}
\end{equation}
The final state of the unmeasured part is
\begin{equation}
|\psi_f\rangle=\Big(\underset{e':k\notin e'}{\prod }C_{e'}^{m_{e'}}\Big)\Big(\underset{e:k\in e}{\prod }C_{e\backslash \{k\}}^{i_k m_e}\Big)|+\rangle ^{\otimes N-1}_d,
\end{equation}
which is also a qudit hypergraph state. Denote the multi-hypergraph as $H_d'=(V',E')$, where $V'=\{1,2,\cdots,k-1,k+1,\cdots,N\}$ and $E'$ is the multiset of the hyperedges of $H_d'$, then the associated multiplicity function is
\begin{equation}
\label{eqn:Xk3}
m_e'= m_e +i_k m_{e\cup \{k\}} \pmod d.
\end{equation}
If the measurement breaks some hyperedges, then we can get the remaining hypergraph through the following steps: (i) Delete the measured vertex; (ii) Multiply the multiplicities of the broken hyperedges by $i_k$ (the measurement result); (iii) Make the ``multiplicities'' valid by subtracting some multiple of $d$.

In general, for a 2-level hypergraph state, after a local Pauli-$X_k$ measurement, the unmeasured part collapses to a state that does not correspond to a hypergraph \cite{Qu13}. For the more general $d$-level hypergraph states, the generalized $X_k$ measurements cannot maintain the structure of the set of $d$-level hypergraph states either.

\black The situations demonstrated in this section are interesting and important because qudit hypergraph states' responses to basic local unitary operations and measurements have potential applications in quantum codes and quantum error correction.
\black

\section{Bell non-locality of qudit hypergraph states and the experimental detection}
\label{sec:5}
The exhibition of non-locality by graph states and qubit hypergraph states
is very important and even necessary in many quantum information tasks.
Behind such investigation is the challenging problem of the non-locality of multipartite
entangled states in quantum information theory. It has been proven
that all entangled pure states are non-local, no matter how many particles
there are and how many dimensions each particle contains \cite{Popescu1992293, Yu2012}.
 In particular, a scheme of non-locality exhibition
was provided in an operational manner in Ref. \cite{Popescu1992293}. \black The idea is that any two particles
can be measured to violate Clauser-Horne-Shimony-Holt (CHSH) inequality \cite{CHSH}, in
assistance of measuring the rest particles. Below we discuss how it
works in the scenario of qudit-hypergraph states.

\subsection{Nonlocality exhibition by the CHSH inequality}

\black
For clarity of the discussion, we consider a multi-hypergraph $H_{N,d,m}=\left(V, E\right)$, in which $V=\{1,2,\cdots,N\}$ and
$E=\{V,V,\cdots,V\}$ ($|E|=m$), then the corresponding quantum state is
\begin{equation}
\begin{split}
\left|H_{N,d,m}\right\rangle  & =C_{V}^{m}\left|+\right\rangle_d ^{\otimes N}\label{eq:hyperex}\\
 & =\frac{1}{\sqrt{d^{N}}}\sum_{i_1,\cdots,i_{N}=0}^{d-1}\omega_{d}^{m i_{1}\cdots i_{N}}\left|i_{1}\cdots i_{N}\right\rangle.
\end{split}
\end{equation}
Without losing generality, let the vertices $3,4,\cdots,N$
assist the vertices $1$ and $2$ in exhibiting Bell non-locality
by violating CHSH inequality \footnote{Notice that for the state $\ket{H_{N,d,m}}$, the
exchange of any pair of vertices does not change the state.}. The assistance can be done by projecting the
vertices to their respective $\left|+\right\rangle_d$. After this operation, the state of the rest part (composed of vertices 1 and 2)
becomes
\begin{equation}
\label{eq:tmp1}
\left|H_{N,d,m}^{(2)}\right\rangle =\frac{\mathscr{N}}{d^{N-1}}\sum_{i_{1},i_{2}=0}^{d-1}\Omega_{i_{1}i_{2}}\left|i_{1}i_{2}\right\rangle,
\end{equation}
where $\mathscr{N}$ is the normalization factor, $\Omega_{i_{1}i_{2}}=\sum_{i_{3},\cdots,i_{N}=0}^{d-1}\omega^{m i_{1}i_{2}\cdots i_{N}}_d$ forming the $d\times d$ matrix $\Omega$. 
%The state $\left|H_{N,d,m}^{(2)}\right\rangle$ in Eq. \eqref{eq:tmp1} is unnormalized, with $\left\Vert\left|H_{N,d,m}^{(2)}\right\rangle\right\Vert^2$ equalling to the probability
%of obtaining the result. 
We state that the remaining two vertices are entangled. The proof can
be done through analyzing the rank of $\Omega$.
Assume $\left|H_{N,d,m}^{(2)}\right\rangle$ is separable, the rank of $\Omega$ should be 1.
But the upper-left $2 \times 2$ submatrix ($i_{1},i_{2}\in\left\{0,1\right\}$) of $\Omega$
\begin{equation}
\tilde{\Omega}=\left(
\begin{array}{cc}
 d^{N-2} & d^{N-2} \\
d^{N-2} & \displaystyle\sum_{i_{3},\cdots,i_{N}=0}^{d-1}\omega^{m i_{3}\cdots i_{N}}\\
\end{array}
\right),
\end{equation}
has a non-zero determinant, therefore $rank(\Omega)\geq2$ \cite{horn2012}, indicating that the rest two
vertices are entangled.

For analyzing the entanglement property and Bell non-locality, it
is more convenient to transform $\left|H_{N,d,m}^{(2)}\right\rangle$ to its Schmidt form
\begin{equation}
\label{eq:tmp2}
\left|H_{N,d,m}^{(2)}\right\rangle =\sum_{\mu=0}^{d-1}c_{\mu}\left|\mu\right\rangle_1 \left|\mu\right\rangle_2,
\end{equation}
where 
%$q=\left\Vert\left|H_{N,d,m}^{(2)}\right\rangle\right\Vert^2$, 
$c_{\mu}$ are the Schmidt coefficients, 
and $\left|\mu\right\rangle_1$ and $\left|\mu\right\rangle_2$ are the Schmidt bases for vertex 1 and 2, respectively.
The entanglement of $\left|H_{N,d,m}^{(2)}\right\rangle$ implies that there are more than 1 non-trivial
term in the right hand side of Eq. \eqref{eq:tmp2}. Finally, we can measure vertex 1 on the settings $S_{1}=\sigma_{z}$ and $T_{1}=\sigma_{x}$, and vertex 2 on the settings $S_{2}=\sigma_{z}\cos2t+\sigma_{x}\sin2t$ and $T_{2}=\sigma_{z}\cos2t-\sigma_{x}\sin2t$, where $\sigma_{z}=\left|0\right\rangle \left\langle 0\right|-\left|1\right\rangle \left\langle 1\right|$, $\sigma_{x}=\left|0\right\rangle \left\langle 1\right|+\left|1\right\rangle \left\langle 0\right|$ on respective basis $\left|\mu\right\rangle_1$ and $\left|\mu\right\rangle_2$, and $\tan2t=2c_{0}c_{1}$. The measurement results will disclose the nonlocality by violating the following CHSH inequality \cite{Popescu1992293}
\begin{align}
C=\bigg | E\left(S_{1}S_{2}\Big|\left|+\right\rangle_d ^{\otimes N-2}\right)+E\left(S_{1}T_{2}\Big|\left|+\right\rangle_d ^{\otimes N-2}\right)\nonumber \\
+E\left(T_{1}S_{2}\Big|\left|+\right\rangle_d ^{\otimes N-2}\right)-E\left(T_{1}T_{2}\Big|\left|+\right\rangle_d ^{\otimes N-2}\right)\bigg | & \leq2.\label{eq:tmp3}
\end{align}
More precisely, the left hand side of the above inequality can achieve $2 \sqrt{1+{4 c_0^2 c_1^2}/{\left(c_0^2 + c_1^2\right)^2}}$, such that the \emph{lhv} bound 2 is violated.

\subsection{The prime-dimensional case}

\black
Specially, when the dimension of the qudits is prime ($d\in\mathbb{P}$), $\Omega_{i_{1}i_{2}}$ has a simple analytic form
\begin{equation}\label{eq:Omega}
\Omega_{i_{1}i_{2}}=\begin{cases}
\begin{array}[t]{cc}
d^{N-2} & i_{1}=0 \lor i_{2}=0,\\
d^{N-2}-d (d-1)^{N-3}  & i_{1} \neq 0 \land i_{2} \neq 0.
\end{array}\end{cases}
\end{equation}
In this case, the Schmidt form of $\left|H_{N,d,m}^{(2)}\right\rangle$ is
\begin{equation}
\ket{H_{N,d,m}^{(2)}}=\frac{x_+\ket{0}_1\ket{0}_2+x_-\ket{1}_1\ket{1}_2}{\sqrt{x_+^2+x_-^2}},
\end{equation}
where
\begin{equation}
x_{\pm}=\frac{\lambda\pm \sqrt{\lambda^2+4(d-\lambda)}}{2},
\end{equation}
and
\begin{equation}
\begin{split}
\ket{0}_k&=\frac{1}{N_+}\Big((x_+ -\lambda+1)\ket{0}+\sum_{i=1}^{d-1}{\ket{i}}\Big),\\
\ket{1}_k&=\frac{1}{N_-}\Big((x_- -\lambda+1)\ket{0}+\sum_{i=1}^{d-1}{\ket{i}}\Big),
\end{split}
\end{equation}
%\begin{equation}
%x_{\pm}=\frac{d-\lambda(d-1)\pm \sqrt{(d-\lambda(d-1))^2+4\lambda (d-1)}}{2},
%\end{equation}
with $N_{\pm}=\sqrt{(x_{\pm} -\lambda+1)^2+d-1}$, $k\in{\{1,2\}}$ and 
$\lambda=d-(d-1)^{N-2}/d^{N-3}$.
The Schmidt number of $\ket{H_{N,d,m}^{(2)}}$ is 2, which indicates that the entanglement of vertices 1 and 2 is equivalent to the entanglement of two qubits. The results in the previous paragraph can be applied here directly except that here $\tan2t=2x_+x_- /\left(x_+^2 + x_-^2\right)$. Explicitly, in this case the left-hand side of Eq. \eqref{eq:tmp3} can violate the CHSH inequality \black by an amount of $2 \sqrt{1+{4 x_+^2 x_-^2}/{\left(x_+^2 + x_-^2\right)^2}}$.

Figure \ref{fig3} reveals the violation of CHSH inequality for various combinations of $d$ $(d\in \mathbb{P})$ and $N$ in this measurement scheme. Here, $C$ is always greater than 2, indicating that this measurement scheme can reveal the nonclassical correlation between the vertices.
When $d$ is fixed and $N$ is large (see (a) in Fig. \ref{fig3}), the matrix elements of the normalized $\Omega$  are \black nearly equal, i.e., the normalized quantum state of the remaining vertices is approximately $|+\rangle_d ^{\otimes 2} $, thus $C$ approaches 2 when $N$ goes to infinity. When $N$ is fixed and $d$ increases (see (b) in Fig. \ref{fig3}), $\ket{H_{N,d,m}^{(2)}}$ approaches $(\ket{0} \sum_{i=1}^{d-1}\ket{i}+\sum_{i=1}^{d-1}\ket{i} \ket{0})/\sqrt{2(d-1)}$, which is equivalent to a 2-qubit maximally entangled state, thus $C$ approaches $2\sqrt{2}$ when $d$ goes to infinity.

\begin{figure}[!htb]
\includegraphics[width =8.5 cm]{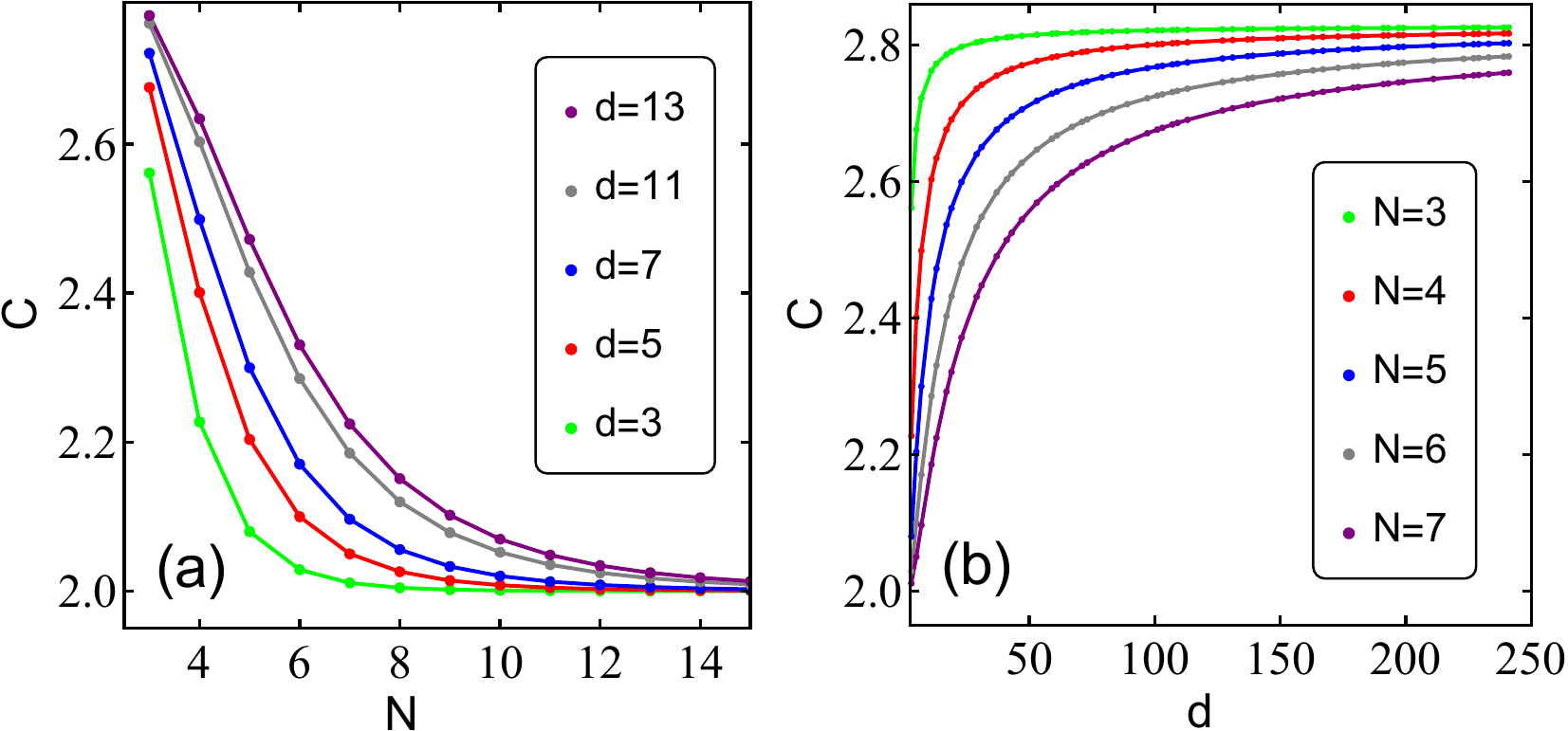}
\caption{Violation of CHSH inequality for various combinations of $d$ and $N$ in our measurement scheme, where $d$ is the dimension of each vertex (notice that $d$ is assumed to be prime) and $N$ is the number of vertices in the multi-hypergraph. The data points are connected for revealing the monotonicity of $C$ with respect to $N$ ($d$). } \label{fig3}
\end{figure}

\subsection{Discussion}

Remarkably, the above scheme of exhibiting the nonlocality of multipartite quantum systems is potentially applicable in the current use of qudit hypergraph states and conventional qubit graph states. In fact, it can be, and in some cases has been, used in practice with current technology. In the case of entanglement verification, it involves only two measurement settings at each side, where the measurement settings of assistant qudits never change. Besides, the CHSH inequality can always reveal the ``strong'' nonlocality, in the sense that the entanglement between arbitrary faraway two qudits can be revealed, as long as the two are connected by other vertices and edges. All these features make the above scheme rather experimentally friendly.

For example, the entanglement verification of cluster states (a special class of graph states) generated by cold-atom lattices is a necessary work for the future use in quantum computing. However, the detection of the entanglement in large-scale cluster states is always a challenging problem \cite{Dai2016}.  From a practical perspective, the CHSH scheme discussed in this section can also be used as an entanglement witness for cluster states, especially for the \textbf{\textit{long-distance}} entanglement. That is, one can always choose two interested particles (connected by other particles with C-phase operations), and test the entanglement correlation between them, no matter how far the two particles are. 

Another example is its application in quantum networks \cite{Cirac1997,Kimble2008}, in which thousands of users complete a quantum information task via a multipartite entangled state.
A typical task is the so-called third-man quantum cryptography in which generation of a cryptographic key is controlled by a third operator who decides whether to activate the key generation or not \cite{Zukowski1998a}. 
% any two users can finally share a maximally entangled state with the help of others.
% This idea is often called open destination in the quantum cryptography context \cite{Zhao2004}. 
The scheme we discussed thus exactly offers an operational way to analyze the security of the third-man quantum cryptography. 

An important problem in qudit hypergraph states we did not discuss is the genuine multipartite entanglement. In particular, the relationship between the classification of multipartite entanglement and the property of hypergraphs deserves to be studied deeply, and the triple entanglement case has been discussed in \cite{steinhoff2016qudit}. As an analog, the concept of genuine multipartite nonlocality is also put forwarded in \cite{steinhoff2016qudit}. However, despite its significance in the theoretical study, its applications in quantum information processing need further study.

\section{Conclusions}
In this work, we have proposed a large class of quantum states named qudit hypergraph states in which every vertex of the multi-hypergraph represents a $d$-level quantum system. We have investigated the operational definition of these states and studied their stabilizers, which possess potential applications in quantum codes and quantum computation. It is shown that generalized local $X$, $Z$ operations, and $Z$ measurements can transform a qudit hypergraph state to another one.

The multi-hypergraphs and qudit hypergraph states have a one-to-one correspondence, and the entanglement of the qudit hypergraph states can be directly illustrated by the structure of their corresponding multi-hypergraphs. If a multi-hypergraph (or part of it) is connected, the corresponding quantum system (the quantum system corresponding to the connected part) is genuinely entangled. Such entanglement leads to potential exhibition of Bell non-locality. As an example, we showed how to obtain the violation of Bell inequality in $N$-uniform qudit hypergraph states. The method is also applicable to other qudit hypergraph states and general $N$-qudit quantum states. 

We also study the relationship among qudit hypergraph states and some important state classes. As for the real equally weighted states, we generalized them to the qudit case according to their form in the computational basis. It is shown that only in the 2-level case are the two state classes (``generalized real equally weighted states'' and ``qudit hypergraph states'') the same; otherwise, qudit hypergraph states are a subclass of ``generalized real equally weighted states''. 
The relationship among qudit hypergraph states, qudit graph states and stabilizer states are discussed. Our results demonstrate that qudit graph states are a common subclass of qudit hypergraph states and stabilizer states. What is more, the union of these two state classes contains more than qudit graph states, which is very similar to the qubit case.
\black

Nevertheless, much work is still needed to be done for the potential properties and applications of qudit hypergraph states. It is known that the set of qubit hypergraph states is the same as the set of \emph{real equally weighted state}s, which is a class of quantum states having important applications in quantum algorithms. Qudit hypergraph states form a subclass of generalized real equally weighted states.  In this sense, it is highly probable that qudit hypergraph states also have important applications in quantum algorithms. It has been shown in the literature that the unique entanglement form and Bell non-locality of qubit hypergraph states have important applications in quantum metrology and novel quantum computation schemes. It is worth further study to see whether the qudit hypergraph states have similar applications. In this paper, we have focused on the simplest definition of entanglement (a quantum state is entangled if it can not be written as a tensor product of two state vectors), while in fact there are much more comprehensive content in the study of multipartite entanglement, for example, equivalent classes of multipartite entanglement. The discussion of such issues in the context of qudit hypergraph states is not only interesting by itself but also essential for the future applications.

\label{sec:6}
\acknowledgements
We thank Y. Y. Zhao, Y. Liu and Y. L. Zheng for the helpful discussions, and F. E. S. Steinhoff for his nice comments. F. L. X. and Z. B. C. were supported by the National Natural Science Foundation of China (Grant No. 61125502) and the CAS. Y. Z. Z., W. F. C., and K. C. were supported by the National Natural Science Foundation of China (Grants No. 11575174) and the CAS.

\vspace{0.2cm}
{\em Note added:} During the preparation of our manuscript, we notice that there is a paper which proposed qudit hypergraph states in a different manner and discussed some different issues \cite{steinhoff2016qudit}.

\begin{appendix}
\section{Derivation of Equation \eqref{eq11}}
If $k \notin e$, $C_e X_k C_e^{\dagger }=X_k$.

If $k \in e$, for simplicity of the discussion, assume that $k=1$ and $e=\{1,\cdots ,n\}$, then $C_e=\sum_{i_1=0}^{d-1}\hat{\Pi}_{i_1} C_{e\backslash \{1\}}^{i_1}$ (Lemma \ref{lem1}), thus
\begin{align}
&C_e X_1 C_e^{\dagger }\nonumber\\
=&\underset{i_1,j_1=0}{\overset{d-1}{\sum }}\hat{\Pi}_{i_1}\bigg(|0\rangle \langle 1|+\ketbra{1}{2}+\cdots +|d-1\rangle \langle 0|\bigg)
\hat{\Pi}_{j_1}C_{e\backslash \{1\}}^{i_1-j_1}\nonumber\\
=&|0\rangle \langle 1|C_{e\backslash \{1\}}^{-1}+|1\rangle \langle 2|C_{e\backslash \{1\}}^{-1}+\cdots +|d-1\rangle \langle 0|C_{e\backslash \{1\}}^{-1}\nonumber\\
=&X_1C_{e\backslash \{1\}}^{-1}=X_1C_{e\backslash \{1\}}^{\dagger}.
\end{align}
Generally, $C_e X_k C_e^{\dagger }=X_k C_{e\backslash \{k\}}^{\dagger}$.

Let $X_k$ pass over all $C_e$, then we have
\begin{equation}
\bigg(\prod _{e\in E} C_e^{m_e}\bigg) X_k \bigg(\prod _{{e'}\in E} C_{e'}^{m_{e'}}\bigg)^{\dagger} =X_k\underset{e:k\in e}{\prod }\left(C_{e\backslash \{k\}}^{\dagger }\right){}^{m_e},
\end{equation}
which is exactly what is demonstrated in Eq. \eqref{eq11}.

\vspace{0.4cm}
\section{Derivation of Equation \eqref{eqn:Xk}}
\label{A2}
For the vertex $e=\{1,2,\cdots,n\}$, denote that $\hat{\Pi}_{i_2 \cdots i_n}=|i_2 \cdots  i_n\rangle \langle i_2 \cdots  i_n|$, then
\begin{align}
&X_1 C_e \nonumber \\
=&X_1\underset{i_2,\cdots ,i_n=0}{\overset{d-1}{\sum }}Z_1^{i_2 \cdots  i_n}\hat{\Pi}_{i_2 \cdots i_n} \nonumber\\
=&\underset{i_2,\cdots ,i_n=0}{\overset{d-1}{\sum }}\omega ^{i_2 \cdots  i_n} Z_1^{i_2 \cdots  i_n} X_1 \hat{\Pi}_{i_2 \cdots i_n}\nonumber\\
=&\underset{i_2,\cdots ,i_n=0}{\overset{d-1}{\sum }}Z_1^{i_2 \cdots  i_n}\hat{\Pi}_{i_2 \cdots i_n}\underset{j_2,\cdots ,j_n=0}{\overset{d-1}{\sum }}\omega ^{j_2 \cdots j_n}\hat{\Pi}_{j_2 \cdots j_n}X_1\nonumber\\
=&C_eC_{e\backslash \{1\}}X_1.
\end{align}

Generally, if $k\in e$, $X_k C_e=C_e C_{e\backslash \{k\}}X_k$, so
\begin{align}
&|\psi_f\rangle\nonumber\\
=&X_k \Big({\prod }C_e^{m_e}\Big)\left|+\right\rangle ^{\otimes N}_d\nonumber\\
=&\Big(\underset{e:k\notin e}{\prod }C_{e}^{m_{e}}\Big) X_k \Big(\underset{e':k\in e'}{\prod }C_{e'}^{m_{e'}}\Big) \left|+\right\rangle ^{\otimes N}_d\nonumber\\
=&\Big(\underset{e:k\notin e}{\prod }C_{e}^{m_{e}}\Big) \Big(\underset{e':k\in e'}{\prod }C_{e'}^{m_{e'}}C_{e'\backslash \left\{k\right\}}^{m_{e'}}\Big) X_k\left|+\right\rangle ^{\otimes N}_d\nonumber\\
=&\Big(\prod C_e^{m_e}\Big)\Big(\underset{e':k\in e'}{\prod }C_{e'\backslash \{k\}}^{m_{e'}}\Big)\left|+\right\rangle ^{\otimes N}_d\nonumber\\
=&\Big(\underset{e':k\in e'}{\prod }C_{e'\backslash \{k\}}^{m_{e'}}\Big)\Big(\prod C_e^{m_e}\Big)\left|+\right\rangle ^{\otimes N}_d\nonumber\\
=&\underset{e:k \in e}{\prod }C_{{e}\backslash \{k\}}^{m_{e}}|H_d\rangle.
\end{align}
\end{appendix}

%%%%%%%%%%%%%%%%%%%%%%%%%%%%%%%%%%%%%%%%
% choose a style
%\bibliographystyle{ieeetr}
%\bibliographystyle{unsrt}
%%%%%%%%%%%%%%%%%%%%%%%%%%%%%%%%%%%%%%%%

%%%%%%%%%%%%%%%%%%%%%%%%%%%%%%%%%%%%%%%%
% choose a .bib file
\bibliographystyle{apsrev4-1}
\bibliography{hypergraph.bib}
%%%%%%%%%%%%%%%%%%%%%%%%%%%%%%%%%%%%%%%%

%\input{hypergraph.bbl}
\end{document}